\documentclass[runningheads]{llncs}
\usepackage{graphics,latexsym}

\usepackage{graphicx,epsfig,color}
\usepackage{graphicx}
\usepackage{enumerate}

\newcommand{\scrod}{\quad\nopagebreak}

\newcommand{\timecomplexity}{O({n(\log\log n)\over
\sum_{i=1}^n a_i})}
\newcommand{\prob}{{\rm Pr}}
\newcommand{\appsum}{{\rm apx\_sum}}

\begin{document}

\date{}

\title{Sublinear Time Approximate Sum via Uniform Random Sampling}
\author{Bin Fu\inst{1}\and Wenfeng Li\inst{2} \and Zhiyong Peng\inst{2}}
 \institute{{Department of Computer
Science\\
 University of Texas-Pan American,
 Edinburg, TX 78539, USA.\\
\email{binfu@cs.panam.edu}} \and {Computer School\\
Wuhan University,
 Wuhan, P. R. China\\
 \email{eyestar\_2008@126.com, peng@whu.edu.cn}} } \maketitle



\begin{abstract} We investigate the approximation for computing
the sum $a_1+\cdots +a_n$ with  an input of a list of nonnegative
elements $a_1,\cdots, a_n$. If all elements are in the range
$[0,1]$, there is a randomized algorithm that can compute an
$(1+\epsilon)$-approximation for the sum problem in time
$\timecomplexity$, where $\epsilon$ is a constant in $(0,1)$. Our
randomized algorithm is based on the uniform random sampling, which
selects one element with equal probability from the input list each
time. We also prove a lower bound $\Omega({n\over \sum_{i=1}^n
a_i})$, which almost matches the upper bound, for this problem.
\end{abstract}


\centerline {{\bf Key words:}  Randomization;   Approximate Sum;
Sublinear Time.}

\section{Introduction}

Computing the sum of a list of elements has many applications. This
problem can be found in the high school textbooks. In the textbook
of calculus, we often see how to compute the sum of a list of
elements, and decide if it converges when the number of items is
infinite. Let $\epsilon$ be a real number at least  $0$. Real number
$s$ is an $(1+\epsilon)$-approximation for the sum problem
$a_1,a_2,\cdots, a_n$ if ${\sum_{i=1}^n a_i\over 1+\epsilon}\le s\le
(1+\epsilon)\sum_{i=1}^na_i$. When we have a huge number of data
items and need to compute their sum, an efficient approximation
algorithm becomes essential. Due to the fundamental importance of
this problem, looking for the sublinear time solution for it is an
interesting topic of research.

A similar problem is to compute the mean of a list of items
$a_1,a_2,\cdots, a_n$, whose mean is defined by
${a_1+a_2+\cdots+a_n\over n}$. Using $O({1\over \epsilon^2}\log
{1\over \delta})$ random samples, one can compute the
$(1+\epsilon)$-approximation for the mean, or decides if it is at
most $\delta$~\cite{Hoefding63}. In~\cite{CanettiEvenGoldreich95},
Canetti, Even, and Goldreich showed that the sample size is tight.
In~\cite{MotwaniPanigrahyXu07}, Motwani, Panigrahy, and Xu showed an
$O(\sqrt{n})$ time approximation scheme for computing the sum of $n$
nonnegative elements.  A priority sampling approach for estimating
subsets were studied
in~\cite{AlonDuffieldLundThorup05,DuffieldLundThorup05,BroderFonturaJosifovskiKumarMotwaniNabarPanigrahyTomkinsXu06}.
Using  different cost and application models,  they tried to build a
sketch so that the sum of any subset can be computed approximately
via the sketch.

We feel the uniform sampling is more justifiable than the weighted
sampling. In this paper, we study the approximation for the sum
problem under both deterministic model and randomized model.  In the
randomized model, we still use the uniform random samplings, and
show how the time is reversely depend on the total sum $\sum_{i=1}^n
a_i$. We also prove a lower bound that matches this time bound. An
algorithm of time complexity $\timecomplexity$ for computing a list
of nonnegative elements $a_1,\cdots, a_n$ in $[0,1]$ can be extended
to a general list of nonnegative elements. It implies an algorithm
of time complexity $O({MC(n)\over \sum_{i=1}^n a_i})$ for computing
a list of nonnegative elements of size at most $M$ by converting
each $a_i$ into ${a_i\over M}$, which is always in the range
$[0,1]$.





\section{Randomized Algorithm for the Sum Problem}

In this section, we present a randomized algorithm for computing the
approximate sum of a list of numbers in $[0,1]$.

\subsection{Chernoff Bounds}\label{chernoff-sec}
 The analysis of our randomized
algorithm often use  the well known Chernoff bounds, which are
described below. All proofs of this paper are self-contained except
the following famous theorems in probability theory.

\begin{theorem}[\cite{MotwaniRaghavan00}]\label{chernoff-theorem}
Let $X_1,\ldots , X_n$ be $n$ independent random $0$-$1$ variables,
where $X_i$ takes $1$ with probability $p_i$. Let $X=\sum_{i=1}^n
X_i$, and $\mu=E[X]$. Then for any $\theta>0$,
\begin{enumerate}
\item $\Pr(X<(1-\theta)\mu)<e^{-{1\over 2}\mu\theta^2}$, and
\item
$\Pr(X>(1+\theta)\mu)<\left[{e^{\theta}\over
(1+\theta)^{(1+\theta)}}\right]^{\mu}$.
\end{enumerate}
\end{theorem}

We follow the proof of Theorem~\ref{chernoff-theorem} to make the
following versions (Theorem~\ref{ourchernoff2-theorem}, and
Theorem~\ref{chernoff3-theorem}) of Chernoff bound for our algorithm
analysis.

\begin{theorem}\label{chernoff3-theorem}
Let $X_1,\ldots , X_n$ be $n$ independent random $0$-$1$ variables,
where $X_i$ takes $1$ with probability at least $p$ for $i=1,\ldots
, n$. Let $X=\sum_{i=1}^n X_i$, and $\mu=E[X]$. Then for any
$\theta>0$,
 $\Pr(X<(1-\theta)pn)<e^{-{1\over 2}\theta^2 pn}$.
\end{theorem}

\begin{theorem}\label{ourchernoff2-theorem}
Let $X_1,\ldots , X_n$ be $n$ independent random $0$-$1$ variables,
where $X_i$ takes $1$ with probability at most $p$ for $i=1,\ldots ,
n$. Let $X=\sum_{i=1}^n X_i$. Then for any $\theta>0$,
$\Pr(X>(1+\theta)pn)<\left[{e^{\theta}\over
(1+\theta)^{(1+\theta)}}\right]^{pn}$.
\end{theorem}

Define $g_1(\theta)=e^{-{1\over 2}\theta^2}$ and
$g_2(\theta)={e^{\theta}\over (1+\theta)^{(1+\theta)}}$. Define
$g(\theta)=\max(g_1(\theta),g_2(\theta))$. We note that
$g_1(\theta)$ and $g_2(\theta)$ are always strictly less than $1$
for all $\theta>0$. It is trivial for $g_1(\theta)$. For
$g_2(\theta)$, this can be verified by checking that the  function
$f(x)=x-(1+x)\ln (1+x)$ is decreasing and $f(0)=0$. This is because
$f'(x)=-\ln (1+x)$ which is strictly less than $0$ for all $x>0$.
Thus, $g_2(\theta)$ is also decreasing, and less than $1$ for all
$\theta>0$.

\subsection{A Sublinear Time Algorithm}

In this section, we show an algorithm to compute the approximate sum
in a sublinear time in the cases that $\sum_{i=1}^n a_i$ is at least
$(\log\log n)^{1+\epsilon}$ for any constant $\epsilon>0$. This is a
randomized algorithm with uniform random sampling.

\begin{theorem}\label{main-thm}
Let $\epsilon$ be a positive constant in $(0,1)$. There is a
sublinear time algorithm such that given  a list of items
$a_1,a_2,\cdots, a_n$ in $[0,1]$, it gives a
$(1+\epsilon)$-approximation in the time $\timecomplexity$.
\end{theorem}

\begin{definition}\label{partition-def}\scrod
\begin{itemize}
\item
For each interval $I$ and a list of items $L$, define $A(I,L)$ to be
the number of items of $L$ in $I$.

\item
For $\delta$, and $\gamma$ in $(0,1)$, a {\it
$(\delta,\gamma)$-partition} for $[0,1]$ divides the interval
$[0,1]$ into intervals $I_1=[\pi_1, \pi_0], I_2=[\pi_2, \pi_1),
I_3=[\pi_3,\pi_2),\ldots ,I_k=[0, \pi_{k-1})$
 such that $\pi_0=1,
\pi_i=\pi_{i-1}(1-\delta)$ for $i=1,2,\ldots , k-1$, and $\pi_{k-1}$
is the first element $\pi_{k-1}\le {\gamma\over n^2}$.

\item
For a set $A$, $|A|$ is the number of elements in $A$. For a list
$L$ of items, $|L|$ is the number of items in $L$.
\end{itemize}
\end{definition}


A brief description of the idea is presented before the formal
algorithm and its proof. In order to get an
$(1+\epsilon)$-approximation for the sum of $n$ input numbers in the
list $L$, a parameter $\delta$ is selected with $1-{\epsilon\over
2}\le (1-\delta)^3$. For a $(\delta,\delta)$-partition $I_1\cup
I_2\ldots \cup I_k$ for $[0,1]$, Algorithm Approximate-Sum$(.)$
below gives the estimation for the number of items in each $I_j$ if
interval $I_j$ has a sufficient number of items. Otherwise, those
items in $I_j$ can be ignored without affecting much of the
approximation ratio. We have an adaptive way to do random samplings
in a series of phases. Let $s_t$ denote the number of random samples
in phase $t$. Phase $t+1$ doubles the number of random samples of
phase $t$ ($s_{t+1}=2s_t$). Let $L$ be the input list of items in
the range $[0,1]$. Let $d_j$ be the number items in $I_j$ from the
samples. For each phase, if an interval $I_j$ shows sufficient
number of items from the random samples, the number of items
$A(I_j,L)$ in $I_j$ can be sufficiently approximated by
$\hat{A}(I_j,L)=d_j\cdot {n\over s_t}$. Thus, $\hat{A}(I_j,L)\pi_j$
also gives an approximation for the sum of the sizes of items in
$I_j$. The sum $\appsum=\sum_{I_j}\hat{A}(I_j,L)\pi_j$ for those
intervals $I_j$ with large number of samples gives an approximation
for the total sum $\sum_{i=1}^na_i$ of the input list.  In the early
stages, $\appsum$ is much smaller than ${n\over s_t}$. Eventually,
$\appsum$ will surpass ${n\over s_t}$. This happens when $s_t$ is
more than ${n\over \sum_{i=1}^n a_i}$ and $\appsum$ is close to the
sum $\sum_{i=1}^n a_i$ of all items from the input list. This
indicates that the number of random samples is sufficient for
approximation algorithm. For those intervals with small number of
samples, their items only form a small fraction of the total sum.
This process is terminated when ignoring all those intervals with
none or small number of samples does not affect much of the accuracy
of approximation. The algorithm gives up the process of random
sampling when $s_t$ surpasses $n$, and switches to use a
deterministic way to access the input list, which happens when the
total sum of the sizes of input items is $O(1)$.

The computation time at each phase $i$ is $O(s_i)$. If phase $t$ is
the last phase, the total time is $O(s_t+{s_t\over 2}+{s_t\over
2^2}+\cdots)=O(s_t)$, which is close to $O({n\over \sum_{i=1}^n
a_i})$. Our final complexity upper bound is $\timecomplexity$, where
$\log\log n$ factor is caused by the probability amplification of
$O(\log n)$ stages and $O(\log n)$ intervals of the
$(\delta,\delta)$ partition in the randomized algorithm.

\vskip 10pt

 {\bf Algorithm Approximate-Sum$(\epsilon,  \alpha, n, L)$}

Input: a parameter, a small parameter $\epsilon\in (0,1)$, a failure
probability upper bound $\alpha$, an integer $n$, a list $L$ of $n$
items $a_1,\ldots , a_n$ in $[0,1]$.

Steps:

\begin{enumerate}[1.]

\item
Phase $0$:
\item\label{parameters-phase0}
\qquad Select $\delta={\epsilon\over 6}$  that satisfies
$1-{\epsilon\over 2}\le (1-\delta)^3$.

\item
\qquad Let $P$ be a $(\delta,\delta)$-partition $I_1\cup I_2\ldots
\cup I_k$ for $[0,1]$.

\item\label{xi0-setting}
\qquad Let $\xi_0$ be a parameter such that $8(k+1)(\log n)
g(\delta)^{(\xi_0\log \log n)/2}<\alpha$ for all large $n$.

\item\label{first-alpha-setting}
\qquad Let $z:=\xi_0\log \log n$.

\item\label{constant-setting-in-Approximate-Intervals}
\qquad Let parameters $c_1:={\delta^2\over 2(1+\delta)}$, and
$c_2:={12\xi_0\over (1-\delta) c_1}$.

\item
\qquad Let $s_0:=z$.

\item
End of Phase $0$.

\item
Phase $t$:


\item\label{loop-m-start}
\qquad Let $s_t:=2s_{t-1}$.

\item
\qquad Sample $s_t$ random items $a_{i_1},\ldots , a_{i_{s_t}}$ from
the input list $L$.

\item
\qquad Let $d_j:=|\{h: a_{i_h}\in I_j\ \ and \ 1\le h\le s_t\}|$ for
$j=1,2,\ldots , k$.

\item\label{I-j-loop}
\qquad For each $I_j$,

\item
\qquad\qquad if $d_j\ge z$,

\item\label{assign-hat-C}
\qquad\qquad then let $\hat{A}(I_j,L):={n\over s_t}d_j$ to
approximate $A(I_j,L)$.

\item\label{I-j-loop-end}
\qquad\qquad else let $\hat{A}(I_j,L):=0$.

\item\label{loop-m-end}
\qquad Let $\appsum:=\sum_{d_j\ge z}\hat{A}(I_j,L)\pi_j$ to
approximate $\sum_{i=1}^n a_n$.

\item\label{until-condition}
\qquad If $\appsum\le {2c_2n\log\log n\over s_t}$ and $s_t< n$ then
enter Phase $t+1$.

\item
\qquad else

\item
\qquad\qquad If $s_t<n$

\item
\qquad\qquad then let $\appsum:=\sum_{d_j\ge z}\hat{A}(I_j,L)\pi_j$
to approximate $\sum_{1\le i\le n}a_i$.

\item
\qquad\qquad else let $\appsum:=\sum_{i=1}^na_i$.

\item
\qquad\qquad Output $\appsum$ and terminate the algorithm.

\item
End of Phase $t$.
\end{enumerate}

{\bf End of Algorithm}

\vskip 10pt

Several lemmas will be proved in order to show the performance of
the algorithm.  Let $\delta,\xi_0, c_1$, and $c_2$ be parameters
defined as those in the Phase 0 of the algorithm
Approximate-Sum$(.)$.

\begin{lemma}\label{prelimary-lemma}\scrod
\begin{enumerate}
\item\label{k-bound}
For parameter $\delta$ in $(0,1)$, a $(\delta,\delta)$-partition for
$[0,1]$ has the number of intervals $k= O({\log n+\log {1\over
\delta}\over \delta})$.

\item\label{g(x)-bound}
$g(x)\le e^{-{x^2\over 4}}$ when $0<x\le {1\over 2}$.

\item\label{x0-bound} The parameter
$\xi_0$ can be set to be $O({\log {1\over \alpha\delta}\over \log
{1\over g(\delta)}})=O({\log {1\over \alpha\delta}\over \delta^2})$
for line~\ref{xi0-setting} in the algorithm Approximate-Sum(.).

\item\label{g(x)-decreasing}
Function $g(x)$ is decreasing and $g(x)<1$ for every $x>0$.
\end{enumerate}
\end{lemma}

\begin{proof} Statement~\ref{k-bound}:
The number of intervals $k$ is the least integer with
$(1-\delta)^k\le {\delta\over n^2}$. We have $k= O({\log n+\log
{1\over \delta}\over \delta})$.

Statement~\ref{g(x)-bound}: By definition
$g(x)=\max(g_1(x),g_2(x))$, where $g_1(x)=e^{-{1\over 2}x^2}$ and
$g_2(x)={e^{x}\over (1+x)^{(1+x)}}$. We just need to prove that
$g_2(x)\le e^{-{x^2\over 4}}$ when $x\le {1\over 2}$. By Taylor
theorem $\ln (1+x)\ge x-{x^2\over 2}$. Assume $0<x\le {1\over 2}$.
We have
\begin{eqnarray*}
\ln g_2(x) &=&x-(1+x)\ln (1+x)\\
&\le&x-(1+x)(x-{x^2\over 2})\\
&=&-{x^2\over 2}(1-x)\\
&\le&-{x^2\over 4}.
\end{eqnarray*}

Statement~\ref{x0-bound}: We need to set up $\xi_0$ to satisfy the
condition in line line~\ref{xi0-setting} in the algorithm. It
follows from statement~\ref{k-bound} and statement~\ref{g(x)-bound}.

Statement~\ref{g(x)-decreasing}: It follows from the fact that
$g_2(x)$ is decreasing, and less than $1$ for each $x>0$. We already
 explained in section~\ref{chernoff-sec}.
\end{proof}

We use the uniform random sampling to approximate the  number of
items in each interval $I_j$ in the $(\delta,\delta)$-partition. Due
to the technical reason, we estimate the failure probability instead
of the success probability.

\begin{lemma}\label{lemma.1}
Let $Q_1$ be the probability that the following statement is false
at the end of each phase:

(i) For each interval $I_j$ with $d_j\ge z$, $(1-\delta)A(I_j,L)\le
\hat{A}(I_j,L)\le (1+\delta)A(I_j,L)$.

Then for each phase in the algorithm, $Q_1\le (k+1)\cdot
g(\delta)^{z\over 2}$.
\end{lemma}

\begin{proof}
 An element of $L$ in $I_j$ is sampled (by an uniform sampling)
with probability $p_j={A(I_j,L)\over n}$. Let $p'={z\over 2s_t}$.
 For each interval $I_j$ with $d_j\ge z$, we discuss two
cases.

\begin{itemize}
\item
Case 1. $p'\ge p_j$.

In this case, $d_j\ge z\ge 2p' s_t\ge 2p_j s_t$. Note that $d_j$ is
the number of elements in interval $I_j$ among $s_t$ random samples
$a_{i_1},\ldots , a_{i_{s_t}}$ from $L$. By
Theorem~\ref{ourchernoff2-theorem} (with $\theta=1$), with
probability at most $P_1=g_2(1)^{p_jm_t}\le g_2(1)^{p's_t}\le
g_2(1)^{z/2}\le g(1)^{z/2}$, there are at least $2p_j s_t$ samples
are from interval $I_j$. Thus, the probability is at most $P_1$ for
the condition of Case 1 to be true.

\item
Case 2. $p'< p_j$.

 By
Theorem~\ref{ourchernoff2-theorem}, we have $\prob[d_j>
(1+\delta)p_jm_t]\le g_2(\delta)^{p_jm_t}\le g_2(\delta)^{p's_t}\le
g_2(\delta)^{z\over 2}\le g(\delta)^{z\over 2}$.


By Theorem~\ref{chernoff3-theorem}, we have $\prob[d_j\le
(1-\delta)p_jm_t]\le g_1(\delta)^ {p_jm_t}\le
g_1(\delta)^{p's_t}=g_1(\delta)^{z\over 2}\le g(\delta)^{z\over 2}$.

For each interval $I_j$ with $d_j\ge z$ and $(1-\delta)p_jm_t\le
d_j\le (1+\delta)p_jm_t$, we have $(1-\delta)A(I_j,L)\le
\hat{A}(I_j,L)\le (1+\delta)A(I_j,L)$ by line~\ref{assign-hat-C} in
Approximate-Sum(.).

There are $k$ intervals $I_1,\ldots , I_k$. Therefore, with
probability at most $P_2=k\cdot g(\delta)^{z\over 2}$,
 the following is false: For each interval $I_j$
with $d_j\ge z$, $(1-\delta)A(I_j,L)\le \hat{A}(I_j,L)\le
(1+\delta)A(I_j,L)$.
\end{itemize}

 By the
analysis of Case 1 and Case 2, we have  $Q_1\le P_1+P_2\le
(k+1)\cdot g(\delta)^{z\over 2}$ (see
statement~\ref{g(x)-decreasing} of Lemma~\ref{prelimary-lemma}).
Thus, the lemma has been proven.
\end{proof}

\begin{lemma}\label{lemma.2}
Assume that $s_t\ge{ c_2n\log\log n\over \sum_{i=1}^n a_i}$. Then
right after executing Phase $t$ in Approximate-Sum$(.)$, with
probability at most $Q_2=2kg(\delta)^{\xi_0\log\log n}$, the
following statement is false:

(ii) For each interval $I_j$ with $A(I_j, L)\ge c_1\sum_{i=1}^n
a_i$, A). $(1-\delta)A(I_j,L)\le \hat{A}(I_j,L)\le
(1+\delta)A(I_j,L)$; and B). $d_j\ge z$.
\end{lemma}

\begin{proof}
Assume that $s_t\ge{ c_2n\log\log n\over \sum_{i=1}^n a_i}$.
Consider each interval $I_j$ with $A(I_j, L)\ge c_1\sum_{i=1}^n
a_i$. We have that $p_j={A(I_j,L)\over n}\ge {c_1\sum_{i=1}^n
a_i\over n}$. An element of $L$ in $I_j$ is sampled with probability
$p_j$.  By Theorem~\ref{ourchernoff2-theorem},
Theorem~\ref{chernoff3-theorem}, and Phase 0 of Approximate-Sum(.),
we have
\begin{eqnarray}
\prob[d_j<(1-\delta)p_jm_t]\le g_1(\delta)^{p_jm_t}\le
g_1(\delta)^{c_1c_2\log\log n}\le g(\delta)^{\xi_0\log\log n}.\\
\prob[d_j>(1+\delta)p_jm_t]\le g_2(\delta)^{p_jm_t}\le
g_2(\delta)^{c_1c_2\log\log n}\le g(\delta)^{\xi_0\log\log n}.
\end{eqnarray}

Therefore, with probability at most $2kg(\delta)^{\xi_0\log\log n}$,
the following statement is false:

For each interval $I_j$ with $A(I_j, L)\ge c_1\sum_{i=1}^n a_i$,
$(1-\delta)A(I_j,L)\le \hat{A}(I_j,L)\le (1+\delta)A(I_j,L)$.

If $d_j\ge (1-\delta)p_j s_t$, then we have
\begin{eqnarray*}
d_j&\ge&  (1-\delta){A(I_j, L)\over n} s_t\\
&\ge& (1-\delta){(c_1\sum_{i=1}^n a_i)\over n}\cdot {
c_2n\log\log n\over \sum_{i=1}^n a_i}\\
&=&(1-\delta)c_1c_2\log\log n\\
&\ge& \xi_0\log\log n= z. \ \ \ \ \ \mbox{(by\ Phase\ 0\ of\
Approximate-Sum(.))}
\end{eqnarray*}
\end{proof}


\begin{lemma}\label{lemma.3}
The total sum of the sizes of items in those $I_j$s with $A(I_j, L)<
c_1\sum_{i=1}^n a_i$ is at most ${ \delta\over 2}(\sum_{i=1}^n
a_i)+{\delta\over n}$.
\end{lemma}

\begin{proof} By Definition~\ref{partition-def}, we have
$\pi_j=(1-\delta)^j$ for $j=1,\ldots ,k-1$. We have that
\begin{itemize}
\item
the sum of sizes of items in $I_k$ is at most $n\cdot {\delta\over
n^2}={\delta\over n}$,
\item
for each interval $I_j$ with $A(I_j, L)< c_1\sum_{i=1}^n a_i$, the
sum of sizes of items in $I_j$ is at most $(c_1\sum_{i=1}^n
a_i)\pi_{j-1}\le (c_1\sum_{i=1}^n a_i)(1-\delta)^{j-1}$ for $j\in
[1,k-1]$.
\end{itemize}
The total sum of the sizes of items in those $I_j$s with $A(I_j, L)<
c_1\sum_{i=1}^n a_i$  is at most
\begin{eqnarray*}
\sum_{j=1}^{k-1}(c_1\sum_{i=1}^n a_i)\pi_{j-1})+\sum_{a_i\in I_k}a_k
&\le& \sum_{j=1}^{k-1}(c_1\sum_{i=1}^n a_i)(1-\delta)^{j-1})+n\cdot
{r\over
n^2}\\
&\le& {c_1\over \delta}(\sum_{i=1}^n a_i)+{\delta\over n}\\
&\le& { \delta\over 2}(\sum_{i=1}^n a_i)+{\delta\over n}.   \ \ \ \
\ \mbox{(by\ Phase\ 0\ of\ Approximate-Sum(.))}
\end{eqnarray*}
\end{proof}

\begin{lemma}\label{lemma.4}
Assume that at the end of phase $t$, for each $I_j$ with
$\hat{A}(I_j,L)>0$, ${A(I_j,L) (1-\delta)}\le \hat{A}(I_j,L)\le
{A(I_j,L)(1+\delta)}$; and $d_j\ge z$ if $A(I_j,L)\ge
c_1\sum_{i=1}^na_i$. Then $(1-{\epsilon\over
2})(\sum_{i=1}^na_i-{4\delta\over n})\le \appsum\le
(1+\delta)(\sum_{i=1}^na_i)$ at the end of phase $t$.
\end{lemma}

\begin{proof} By the assumption of the lemma, we have $\appsum=\sum_{d_j\ge z}\hat{A}(I_j,L)\pi_j\le
(1+\delta)\sum_{i=1}^n a_i$. For each interval $I_j$ with $j\not=k$,
we have $A(I_j,L)\pi_j\ge (1-\delta)\sum_{a_i\in I_j}a_i$ by the
definition of $(\delta,\delta)$-partition.  Thus,
\begin{eqnarray}
A(I_j,L)\pi_j\ge (1-\delta)\sum_{a_i\in I_j}a_i\ \ \ \mbox{for \ \
$j\not=k$.}\label{lower-bound-one-interval}
\end{eqnarray}

By the condition of this lemma and Lemma~\ref{lemma.3},we have
\begin{eqnarray}
\sum_{d_j<z}\sum_{a_i\in I_j} a_i\le { \delta\over 2}(\sum_{i=1}^n
a_i)+{\delta\over n}\ \ \ \ \ \label{lemma3-ineqn}
\end{eqnarray}

 We have the following inequalities:
\begin{eqnarray*}
\appsum&=&\sum_{d_j\ge z}\hat{A}(I_j,L)\pi_j\ \ \ \mbox{(by\ line~\ref{until-condition}\ in\ Approximate-Sum(.))}\\
&\ge&(1-\delta)\sum_{d_j\ge z}A(I_j,L)\pi_j\\
&\ge&(1-\delta)\sum_{d_j\ge z, j\not=k}A(I_j,L)\pi_j\\
&\ge&(1-\delta)^2\sum_{d_j\ge z,j\not=k}\left(\sum_{a_i\in I_j}a_i\right)\ \ \ \ \mbox{(by\ inequality~(\ref{lower-bound-one-interval}))}\\
&\ge&(1-\delta)^2(\sum_{i=1}^na_i-\sum_{d_j<z}\sum_{a_i\in I_j} a_i-\sum_{a_i\in I_k}a_i)\\
&\ge&(1-\delta)^2(\sum_{i=1}^na_i-({ \delta\over 2}(\sum_{i=1}^n a_i)+{\delta\over n})-n\cdot {\delta\over n^2})\ \ \ \mbox{(by\ inequality\ (\ref{lemma3-ineqn}))}\\
&\ge&(1-\delta)^3(\sum_{i=1}^na_i-{4\delta\over
n})\\
&\ge&(1-{\epsilon\over 2})(\sum_{i=1}^na_i-{4\delta\over n}). \ \ \
\mbox{(By\ line~\ref{parameters-phase0}\ in\ Phase\ 0\ of\ the\
algorithm)}
\end{eqnarray*}
\end{proof}


\begin{lemma}\label{lemma.5}
With probability at most $Q_5=(k+1)\cdot (\log n) g(\delta)^{z\over
2}$,  at least one of the following statements is false:
\begin{enumerate}[A.]
\item\label{clm5-a}
For each phase $t$ with $s_t<{c_2n\log\log n\over \sum_{i=1}^n
a_i}$,  the condition $\appsum\le {2c_2n\log\log n\over s_t}$ in
line~\ref{until-condition} of the algorithm is true.
\item\label{clm5-b}
If $\sum_{i=1}^na_i\ge 4$, then the algorithm stops some phase $t$
with $s_t\le {16c_2n\log\log n\over \sum_{i=1}^n a_i}$.

\item\label{clm5-c}
If $\sum_{i=1}^na_i< 4$, then it stops  at a phase $t$ in which the
condition $s_t\ge n$ first becomes true, and outputs
$\appsum=\sum_{i=1}^n a_i$.
\end{enumerate}
\end{lemma}

\begin{proof}
By Lemma~\ref{lemma.1}, with probability at most $(k+1)\cdot
g(\delta)^{z\over 2}$, the statement i of Lemma \ref{lemma.1} is
false for a fixed $m$. The number of phases is at most $\log n$
since $s_t$ is double at each phase. With probability $(k+1)\cdot
(\log n)\cdot g(\delta)^{z\over 2}$, the statement i of Lemma
\ref{lemma.1} is false for each phase $t$ with $s_t\le n$. Assume
that statement i of Lemma \ref{lemma.1} is true for every phase $t$
executed by the algorithm Approximate-Sum(.).

Statement~\ref{clm5-a}.  Assume that $s_t<{c_2n\log\log n\over
\sum_{i=1}^n a_i}$. We have ${n\over s_t}>{n\over {c_2 n\log\log
n\over \sum_{i=1}^n a_i}}={\sum_{i=1}^n a_i\over c_2\log\log n}$.
Therefore, $\sum_{i=1}^n a_i<({n\over s_t})c_2\log\log
n={c_2n\log\log n\over s_t}$.

Since statement i of Lemma \ref{lemma.1} is true,  the condition of
Lemma~~\ref{lemma.4} is satisfied.  By Lemma~\ref{lemma.4},
$\appsum\le (1+\delta)\sum_{i=1}^n a_i$.
 Since $(1+\delta)<2$ (by line~\ref{constant-setting-in-Approximate-Intervals} in Approximate-Sum(.)), we have
\begin{eqnarray*}
 \appsum\le (1+\delta)\sum_{i=1}^n a_i\le 2\sum_{i=1}^n a_i
 < 2\cdot {c_2n\log\log n\over s_t}= {2c_2n\log\log n\over s_t}.
\end{eqnarray*}
Statement~\ref{clm5-b}. The variable $s_t$ is doubled in each new
phase.

Assume that the algorithm enters phase $t$ with ${8c_2n\log\log
n\over \sum_{i=1}^n a_i}\le s_t\le {16 c_2n\log\log n\over
\sum_{i=1}^n a_i}$.  We have
\begin{eqnarray}
{n\over s_t}\le {n\over {8c_2 n\log\log n\over \sum_{i=1}^n
a_i}}={\sum_{i=1}^n a_i\over 8c_2\log\log n}.\label{n-mt-ineqn}
\end{eqnarray}
 Since
$\sum_{i=1}^na_i\ge 4$, $(\sum_{i=1}^na_i-{4\delta\over n})\ge
(1-\delta)(\sum_{i=1}^na_i)$.

By Lemma~\ref{lemma.4}, we have the inequality
\begin{eqnarray}
\appsum\ge (1-{\epsilon\over
2})(1-\delta)(\sum_{i=1}^na_i).\label{appsum-lower-ineqn}
\end{eqnarray}
By the setting at Phase 0 of the algorithm, we have
\begin{eqnarray}
(1-{\epsilon\over 2})(1-\delta)\ge {1\over 2}\cdot {3\over 4}=
{3\over 8}.\label{over-c0-ineqn}
\end{eqnarray}
We have
\begin{eqnarray}
 \appsum&\ge&(1-{\epsilon\over 2})(1-\delta)(\sum_{i=1}^na_i)\ \ \ \ \ \mbox{(by\ inequality~(\ref{appsum-lower-ineqn}))}\\
 &\ge&(1-{\epsilon\over 2})(1-\delta)({n\over s_t}\cdot 8c_2\log\log n)\ \ \ \ \ \mbox{(by\ inequality~(\ref{n-mt-ineqn}))}\\
&\ge&{3\over 8}({n\over s_t}\cdot 8c_2\log\log
n)\\
 &=& {3c_2n\log\log n\over s_t}. \ \ \ \ \ \mbox{(by\
 inequality~(\ref{over-c0-ineqn}))}
\end{eqnarray}
Thus, it makes the condition at line~\ref{until-condition} in
Approximate-Sum(.) be false. Thus, the algorithm stops at some stage
$t$ with $s_t\le {16 c_2n\log\log n\over \sum_{i=1}^n a_i}$ by the
setting at line~\ref{until-condition} in Approximate-Sum(.).

 Statement~\ref{clm5-c}. It follows from statement A and the
setting in line~\ref{until-condition} of the algorithm.
\end{proof}

\begin{lemma}\label{lemma.6}
The complexity of the algorithm is $O({\log {1\over
\alpha\delta}\over \delta^4}\min({n\over \sum_{i=1}^n
a_i},n)\log\log n)$. In particular, the complexity is
$O(\min({n\over \sum_{i=1}^n a_i},n)\log\log n)$ if $\alpha$
 is fixed in $(0,1)$.
\end{lemma}

\begin{proof} We check the size $s_t$ of random samplings according by statement~\ref{clm5-b} and statement~\ref{clm5-c} of
Lemma~\ref{lemma.5} to determine when to stop the algorithm. We have
$\xi_0=O({\log {1\over \alpha\delta}\over \delta^2})$ by
Lemma~\ref{prelimary-lemma}. By the setting in
line~\ref{constant-setting-in-Approximate-Intervals} in
Approximate-Sum(.), we have
\begin{eqnarray*}
c_2&=& {12\xi_0\over (1-\delta)c_1} =O({\log {1\over
\alpha\delta}\over \delta^4}).
\end{eqnarray*}

 Since $s_i$ is doubled every phase,
and each  phase $i$ costs $O(s_i)$ time. The total time of the
algorithm is $O(s_1+s_2+\cdots +s_t)=O(s_t)$, where phase $t$ is the
last phase.

 The computational time complexity
   of the algorithm follows from statement~\ref{clm5-b} and statement~\ref{clm5-c} of
Lemma~\ref{lemma.5}.
\end{proof}

\begin{lemma}\label{app-sum-lemma} With probability at most
$\alpha$, at least one of the following statements is false after
executing the algorithm Approximate-Sum$(\epsilon,  \alpha, n, L)$:
\begin{enumerate}[1.]
\item\label{item3b-app-sum-lemma}
If $\sum_{i=1}^na_i\ge 4$, then $(1-\epsilon)(\sum_{i=1}^na_i)\le
\appsum\le (1+{\epsilon\over 2})(\sum_{i=1}^na_i)$;
\item\label{item3c-app-sum-lemma}
If $\sum_{i=1}^na_i< 4$, then $\appsum=\sum_{i=1}^na_i$; and
\item\label{item4-app-sum-lemma}
It runs in $O({\log {1\over \alpha\delta}\over \delta^4}\min({n\over
\sum_{i=1}^n a_i},n)\log\log n)$ time. In particular, the complexity
of the algorithm is $O(\min({n\over \sum_{i=1}^n a_i},n)\log\log n)$
if $\alpha$ is fixed in $(0,1)$.
\end{enumerate}
\end{lemma}

\begin{proof}
As $s_t$ is doubled each new phase in Approximate-Intervals$(.)$,
the number of phases is at most $\log n$. With probability at most
$(\log n)(Q_1+Q_2)+Q_5\le \alpha$ (by line~\ref{first-alpha-setting}
in Approximate-Intervals$(.)$), at least one of the statements (i)
in Lemma \ref{lemma.1}, (ii) in Lemma \ref{lemma.2}, A, B, C in
Lemma \ref{lemma.5} is false.


Assume that the statements (i) in Lemma \ref{lemma.1}, (ii) in Lemma
\ref{lemma.2}, A, B, and C in Lemma \ref{lemma.5} are all true.

 Statement~\ref{item3b-app-sum-lemma}: The condition of Statement~\ref{item3b-app-sum-lemma} implies $n\ge
4$. By Lemma \ref{lemma.4},  we have
\begin{eqnarray}
(1-{\epsilon\over 2})(\sum_{i=1}^na_i-{4\delta\over n})\le
\appsum\le (1+\delta)(\sum_{i=1}^na_i); \label{appsum-range-ineqn}
\end{eqnarray}

 Since
$\sum_{i=1}^na_i\ge 4$, we have
\begin{eqnarray}
(\sum_{i=1}^na_i-{4\delta\over n})\ge
(1-\delta)(\sum_{i=1}^na_i).\label{middle-ineqn}
\end{eqnarray}
We have the inequality
\begin{eqnarray}
\appsum&\ge& (1-{\epsilon\over
2})(1-\delta)(\sum_{i=1}^na_i)\ \ \ \ \ \mbox{(by\ inequalities\ (\ref{middle-ineqn})\ and\ (\ref{appsum-range-ineqn}))}\\
&\ge& (1-\epsilon)(\sum_{i=1}^na_i). \ \ \ \ \ \mbox{(by Phase 0 in
Approximate-Sum(.))}
\end{eqnarray}

 Statement~\ref{item3c-app-sum-lemma} follows from
Statement~\ref{clm5-c} of Lemma~\ref{lemma.5}.

Statement~\ref{item4-app-sum-lemma} for the running time follows
from Lemma \ref{lemma.6}.

 Thus,
with probability at  most $\alpha$, at least one of the
statements~\ref{item3b-app-sum-lemma} to \ref{item4-app-sum-lemma}
is false.
\end{proof}

Now we have the proof for our main theorem.

\begin{proof}[for Theorem~\ref{main-thm}]
Let $\alpha={1\over 4}$ and $\epsilon\in (0,1)$.  It follows from
Lemma~\ref{app-sum-lemma} via a proper setting for those parameters
in the algorithm Approximate-Sum(.).

 The $(\delta,\delta)$-partition $P:$ $I_1\cup
I_2\ldots \cup I_k$ for $[0,1]$ can be generated in $O({\log n+\log
{1\over \delta}\over \delta})$ time by Lemma~\ref{prelimary-lemma}.
Let $L$ be a list of $n$ numbers in $[0,1]$. Pass $\delta, \alpha,
P, n,$ and $L$ to Approximate-Sum(.), which returns an approximate
sum $\appsum$.

 By statement~\ref{item3b-app-sum-lemma} and
statement~\ref{item3c-app-sum-lemma} of Lemma~\ref{app-sum-lemma},
we have an $(1+\epsilon)$-approximation for the sum problem with
failure probability at most $\alpha$. The computational time is
bounded by $O({\log {1\over \alpha\delta}\over \delta^4}\min({n\over
\sum_{i=1}^n a_i},n)\log\log n)$ by
statement~\ref{item4-app-sum-lemma} of  Lemma~\ref{app-sum-lemma}.
\end{proof}

\begin{definition} Let $f(n)$ be a function from $n$ to $(0,n]$ and a parameter $c>1$.
Define $\sum(c,f(n))$ be the class of sum problem with an input of
nonnegative numbers $a_1,\cdots, a_n$ with $\sum_{i=1}^na_i\in
[{f(n)\over c}, cf(n)]$.
\end{definition}

\begin{corollary} Assume that $f(n)$ is a function from $n$ to $(0,n]$ and $c$ is a given constant $c$ greater than $1$.
There is a $O({n (\log\log n)\over f(n)})$ time algorithm such that
given  a list of nonnegative numbers $a_1,a_2,\cdots, a_n$ in
$\sum(c, f(n))$, it gives a $(1-\epsilon)$-approximation.
\end{corollary}

\begin{proof}
It follows from Theorem~\ref{main-thm}.
\end{proof}

We can extend our sublinear time algorithm to the more general list
of nonnegative elements.

\begin{theorem}
Assume that $\epsilon$ is a positive constant in $(0,1)$. Then there
is an $O({Mn (\log\log n)\over \sum_{i=1}^n a_i})$ time algorithm to
compute $(1+\epsilon)$-approximation for a list of nonnegative
numbers $a_1,\cdots, a_n$ of in the range $[0,M]$.
\end{theorem}

\begin{proof}
A list of nonnegative elements $a_1,\cdots, a_n$ can be converted
into the list ${a_1\over M},\cdots, {a_n\over M}$ in $[0,1]$. It
follows from Theorem~\ref{main-thm}.
\end{proof}

\section{Lower Bound }

We show a lower bound for those sum problems with bounded sum of
sizes $\sum_{i=1}^n a_i$. The lower bound always matches the upper
bound.

\begin{theorem}\label{strong-lower-bound-theorem}  Assume $f(n)$ is
an nondecreasing unbounded function from $N$ to $N$ with
$f(n)=o(n)$. Every randomized $(\sqrt{c}-\epsilon)$-approximation
algorithm for the sum problem in $\sum(c, f(n))$ needs
$\Omega({n\over f(n)})$ time, where $c$ is a constant greater than
$1$, and $\epsilon$ is an arbitrary small constant in
$(0,\sqrt{c}-1)$.
\end{theorem}

\begin{proof}
The first list $L_1$ contains $f(n)$ elements of size ${1\over c}$,
and its rest $n-f(n)$ items are  $0$. The sum of numbers in the
first list is ${f(n)\over c}$. Therefore, the first list is a sum
problem in $\sum(c,f(n))$.

The second list $L_2$ contains $f(n)$ elements of value $1$, and its
rest $n-f(n)$ items are  $0$. The sum of numbers in the second list
is $f(n)$. Therefore, the second list is a sum problem in
$\sum(c,f(n))$.

Assume that an algorithm only has computational time $o({n\over
f(n)})$ for computing $k$-approximation for sum problems in $\sum(c,
f(n))$ with $k=(\sqrt{c}-\epsilon)$. For each uniform random
sampling, with probability ${f(n)\over n}$, it gets an number
greater than $0$ in each $L_i$. The algorithm has an $o(1)$
probability to access at least one item greater than $0$ in each
list in a path of computation. Therefore,  $L_1$ and $L_2$ have the
same output for approximation by the same randomized algorithm.  If
$s$ is a $k$-approximation for the both sum problems, we have
\begin{eqnarray}
{f(n)\over ck}&\le& s\le {kf(n)\over c},\ \ \ \ \ \mbox{and} \\
{f(n)\over k}&\le& s\le kf(n)
\end{eqnarray}

We have ${kf(n)\over c}\ge {f(n)\over k}$ for $k=\sqrt{c}-\epsilon$.
This brings  a contradiction.

\end{proof}

\begin{corollary}
There is no $o({n\over \sum_{i=1}^n a_i})$ time randomized
approximation scheme algorithm for the sum problem.
\end{corollary}

\section{Conclusions}

We studied the approximate sum in a few models. We show that the
approximate sum can be computed in time $\timecomplexity$ if the
input list in the range $[0,1]$. Our lower bound almost matches the
upper bound. An interesting theoretical problem is to close the
small gap between the lower bound and upper bound for the
approximate sum problem.


\end{document}